\documentclass[letterpaper,twocolumn,10pt]{article}

\usepackage[T1]{fontenc}
\usepackage[utf8]{inputenc}
\usepackage[english]{babel}

\usepackage{usenix,epsfig}

\usepackage[final]{microtype}
\usepackage{url}
\usepackage[
    doi=false,
    isbn=false,
    url=false,
    maxnames=2,
]{biblatex}
    \AtEveryBibitem{\clearfield{urlyear}}
\bibliography{zotero.bib}
\usepackage{subcaption}
\usepackage{csquotes}
\usepackage{amsthm}

\usepackage{tikz}
\usetikzlibrary{positioning,arrows.meta,matrix,patterns}
\tikzset{font=\footnotesize}
\usepackage{pgfplots}
    \pgfplotsset{
        compat=1.14,
        grid style={dotted},
        grid=major,
        table/search path={.},
        typeset ticklabels with strut,
    }
\usepgfplotslibrary{dateplot,statistics,units}

\usepackage{clrscode3e}
\newcommand{\In}{\ifmmode\ \textrm{\textbf{in}}\ \else\textbf{in}\ \fi}
\newcommand{\Or}{\ifmmode\ \textrm{\textbf{or}}\ \else\textbf{or}\ \fi}
\newcommand{\And}{\ifmmode\ \textrm{\textbf{and}}\ \else\textbf{and}\ \fi}
\newcommand{\Not}{\ifmmode\ \textrm{\textbf{not}}\ \else\textbf{not}\ \fi}
\newcommand{\Continue}{\ifmmode\ \textrm{\textbf{continue}}\ \else\textbf{continue}\ \fi}
\newcommand{\Alert}{\ifmmode\ \textrm{\textbf{alert}}\ \else\textbf{alert}\ \fi}
\renewcommand{\gets}{ $ \leftarrow $ }

\usepackage{newfloat}
\DeclareFloatingEnvironment[
        name=Algorithm,
    ]{algorithm}

\DeclareFloatingEnvironment[
        name=Protocol,
    ]{protocol}

\newcounter{msgno}

\tikzset{>=latex}
\tikzset{
    flowmsg/.style={
        -latex,
    },
}

\usepackage{hyperref}  
\usepackage[
    backgroundcolor=orange!10,
    bordercolor=orange!10,
    linecolor=orange!40,
    textsize=small,
]{todonotes}  

\date{}

\title{\Large \bf Software Distribution Transparency and Auditability}

\author{
    {\rm Benjamin Hof \qquad Georg Carle} \\
    Technical University of Munich
} 

\begin{document}
\maketitle


\subsection*{Abstract}

%

A large user base relies on software updates provided through package
managers.
This provides a unique lever for improving the security of the software update
process.
We propose a transparency system for software updates and implement it for a
widely deployed Linux package manager, namely APT.
Our system is capable of detecting targeted backdoors without producing
overhead for maintainers.
In addition, in our system, the availability of source code is
ensured,
the binding between source and binary code is verified using reproducible
builds, and
the maintainer responsible for distributing a specific package can be
identified.
We describe a novel “hidden version” attack against current software
transparency systems and propose as well as integrate a suitable defense.
%
To address equivocation attacks by the transparency log server, we introduce
tree root cross logging, where the log's Merkle tree root is submitted into a
separately operated log server.
This significantly relaxes the inter-operator cooperation requirements compared
to other systems.
Our implementation is evaluated by replaying over 3000 updates of the Debian
operating system over the course of two years, demonstrating its viability and
identifying numerous irregularities.

\section{Introduction}


Software systems require regular updates.
The protection of software distribution from manipulation is therefore an
integral part of computer security~\cite{
    acar_sok:_2016,
    barrera_baton:_2014,
    barrera_understanding_2012,
    cappos_look_2008,
    fahl_hey_2014,
    kuppusamy_diplomat:_2016,
    karthik_uptane:_2016,
    nikitin_chainiac:_2017,
    samuel_survivable_2010,
}.
The attractiveness of using software updates to distribute malware is
evidenced by several recent attacks piggybacking on legitimate software to
target large enterprises with backdoors~\cite{goodin_ccleaner_2017,
greenberg_software_2017, kaspersky_lab_shadowpad:_2017}.

Many systems rely on package-based software updaters to provide updates, such
as installations of Linux distributions.
These distributions offer a central and collectivized point of security
updates.
Many organizations, unable to provide their own security support for software,
depend on the distributions for important updates.
Due to their central position, distributions consequently provide an important
lever for improving the security of software dissemination.
This distribution process poses a number of challenges.

Package managers generally use cryptographic signatures to protect the
distribution of software packages, if they provide any security at
all~\cite{cappos_look_2008}.
Systems based on signatures are vulnerable to \emph{targeted backdoors}, where
an attacker is able to correctly sign a manipulated version of the code.
The manipulated version is only offered selectively to the
victim~\cite{fahl_hey_2014}.

To prepare software for redistribution, modifications and additions are often
required as integration glue.
The \emph{attribution} of these changes is important, because authorizing
a software for distribution constitutes a statement of confidence into its
benevolence.

Providing assurance of the \emph{mapping between source code and binary code}
is being addressed in the Reproducible Builds
efforts~\cite{levsen_overview_2017}.
Building on this property, a secure software distribution can assure that for
any binary the corresponding \emph{source code is guaranteed to be available}
to assist audit and forensic activities.

In this paper, we develop and evaluate a system to address these challenges
for package managers.
We propose that the software package meta data and source code are submitted
into an untrusted append-only Merkle tree log.
The design of the log allows third parties to efficiently verify its honest
operation.
In detail, our contributions are: 

\begin{itemize}
    \item Protection against targeted backdoors, including a novel ``hidden
        version'' attack against existing software transparency systems
        (Section~\ref{sse:design})
    \item Auditability, providing the inspectable source code corresponding to
        any installed binary and identifying the authorizing maintainer
                (Section~\ref{sse:monitor})
    \item Practical protection against equivocation, reducing
        inter-operator requirements (Section~\ref{sse:equivocation})
    \item Evaluation on over 3000 real Debian distribution updates,
        demonstrating the viability of our system and discovering a
        substantial number of bugs (Section~\ref{sse:evaluation})
    \item Comparison to related work, highlighting the practicability of our
        solution (Section~\ref{sse:comparison})
\end{itemize}


\noindent
We discuss Background in Section~\ref{sse:background} and outline Related
Work in Section~\ref{sse:related}.

\section{Background}
\label{sse:background}

In order to effectively secure the distribution process of software it is
necessary to understand how code is commonly redistributed.
In this section we therefore describe the architecture of the APT package manager.

\subsection{Software distribution models}

For many software projects, a dissemination model along the following lines
applies.
Programmers, referred to as “upstream” in our context, provide
their software for reuse.
They upload the code to code hosting platforms, such as Github, or
programming language specific package managers.
From there, software is downloaded by distribution maintainers for packaging
and integration.
This dissemination model is shown in Figure~\ref{fig:xform}.
Steps marked with boxes regularly include transformations or modifications,
such as packaging or adding meta data.
The parts relevant to the distribution package manager are shown below the
dashed line.

Linux distributions, as an outcome of non-commercial or commercial collaboration, 
have an important role in the distribution of software.
The core task of a distribution is the integration work of making up to tens
of thousands of individual software projects co-installable and standardizing
interfaces for administration.
After integrating a package, distributions regularly provide security support
for it.
Other tasks of the distribution include license vetting and quality assurance.

Distributions are therefore important parts of software dissemination.
Viewing the process of software distribution as a series of transformations, we
can observe that it would ideally be possible for the user to determine the
provenance of each piece of code in reverse direction of the arrows in
Figure~\ref{fig:xform}.
Each of the transformations should be inspectable, by offering a machine
readable description allowing to reproduce the transformation.

\subsection{Debian}

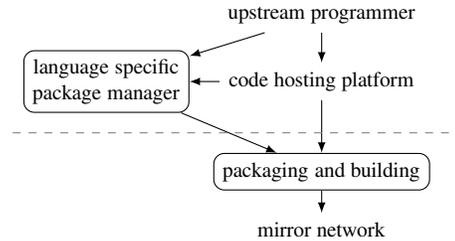
\begin{figure}[t]
    \centering
    \begin{tikzpicture}[
            action/.style={
                draw,
                rounded corners,
            },
            send/.style={
                ->,
            },
        ]

    \node[name=upstream] { upstream programmer };
    \node[name=hoster, below=4mm of upstream] { code hosting platform };
    \node[name=community, left=4mm of hoster, align=center, action]
        { language specific \\ package manager };
    \node[name=packaging, below=7mm of hoster, action]
        { packaging and building };
    \node[name=distribution, below=3mm of packaging] { mirror network };

    \node[name=l, below left=2mm of community] {};
    \node[name=r, right=60mm of l] {};

    \draw[send] (upstream) to (hoster);
    \draw[send] (upstream) to (community);
    \draw[send] (hoster) to (packaging);
    \draw[send] (hoster) to (community);
    \draw[send] (community) to (packaging);
    \draw[send] (packaging) to (distribution);

    \draw[dashed, gray] (l) to (r);

    \end{tikzpicture}
    \caption{Transformations of a software project.}
    \label{fig:xform}
\end{figure}

Debian is one of the oldest Linux distributions, and constitutes the basis for
many derivative distributions~\cite{distrowatch}.
In the following, we provide a simplified overview on how packages are
distributed with the APT package manager in Debian.

On a Debian machine, additional software can be installed by downloading and
installing packages containing the libraries, executables and additional
artifacts using the Advanced Package Tool (APT) package manager.

These packages are created by the maintainers of the Debian distribution, who
upload the signed packages to a central server, the archive.
This server coordinates the building of binary packages for the supported CPU
architectures.
It serves the authoritative copy of the package archive to the repository
mirrors.

To decide no acceptance of a submitted package, the archive verifies that its
signature was created by an authorized uploader key.
This list of keys has the role of an ACL, 
where some people may upload only
specific packages, while others may be allowed to upload any package.
There are other issues that may prevent acceptance of the package, for example
the license under which the software is distributed.
The package is then built for all supported architectures.
Information about the build environment, captured in “buildinfo” files, is
recorded by the build servers~\cite{bobbio_buildinfofiles_2016,
luo_buildinfoinfrastructure_2016}.
If building is successful, the package may be included into the upcoming
release.

The binary and source packages are included into a package index.
These indices (“Packages” and “Sources”) cover all package contents by their
cryptographic hash.
They also contain meta data such as dependencies.
All indices are covered, again by hash, by a release file which is signed.
The release file is signed by the publishing archive, constituting the
main security feature of plain APT.
Indices, release file, and the packages themselves can now be released by the
archive and distributed to the global mirror network, which acts as content
distribution network.
An APT client can retrieve a package from a repository mirror server and
verify its authenticity by confirming its inclusion into the signed release
file.
A release file has a wall clock validity time of two weeks.

\section{Related work}
\label{sse:related}

We consider research on the security of package managers and on Merkle
tree-based log systems.

\vspace{-3mm}
\paragraph{Security of package managers.}
Cappos et al.~\cite{cappos_look_2008} analyze the security of popular package
managers, among those APT.
They define a threat model focused mainly on adversarial control of a package
mirror.
Expanding on this threat model, Samuel et al.\ consider compromise of signing
keys in the design of The Update Framework (TUF), a secure application
updater~\cite{samuel_survivable_2010}.
To guard against key compromise, TUF introduces a number of different roles in
the update release process, each of which operates cryptographic signing keys.

The following three properties are protected by TUF.
The content of updates is secured, meaning its \emph{integrity} is preserved.
Securing the \emph{availability} of updates protects against freeze attacks,
where an outdated version with known vulnerabilities is served in place of a
security update.
The goal of maintaining the correct combination of updates implies the
\emph{security of meta data}.

An attack is deemed successful under the TUF threat model in either of the
following two cases.  The client installs a different software than the most
current version of the software to be updated.  An attack is also successful,
if the client does not install the most current version, but perhaps leaves an
older version installed, without causing an alert.

Later works focus on role separation and adoption to domain specific threat
models~\cite{kuppusamy_diplomat:_2016, karthik_uptane:_2016}.



\vspace{-3mm}
\paragraph{Formal analysis of transparency logs.}
Dowling et al.~\cite{dowling_secure_2016} as well as Chase and
Meiklejohn~\cite{chase_transparency_2016} formally analyze transparency logs.
The model of Chase and Meiklejohn, the transparency overlay, includes
equivocation.
They prove several cryptographic security properties for this transparency,
which can be instantiated to represent Certificate Transparency as well as
Bitcoin.

In a transparency overlay, the dynamic list commitment~(DLC) is defined as a
commitment to a list of elements that can represent exactly one list e.g., a
Merkle tree root.
The list of elements can only be updated by appending elements, producing a
new commitment for the new list.
The commitments allows to efficiently prove that the list has been operated
append-only, and that an element is part of the list.
Notable properties include the security of cryptographic evidence.
For any kind of violation, there exists evidence provably identifying the
culpable party.
This evidence is infeasible to fabricate.

The overlay extends an abstract system with the roles of log, auditor, and
monitor.
The log stores the events produced by the system, its DLC can be instantiated
with a Merkle tree.
The auditor verifies consistency and inclusion of events without having to
store the entire list maintained by the log.
The monitor retrieves elements from the log, verifies log consistency, and
analyses the new events in order to flag any entries that are considered
problematic.
By exchanging observed commitments, the auditor and monitor can detect
equivocation by the log.

\vspace{-3mm}
\paragraph{Securing package updates with co-signing.}
Nikitin et al.\ develop CHAINIAC, a system for software update
transparency~\cite{nikitin_chainiac:_2017}.
Software developers create a Merkle tree over a software package and the
corresponding binaries.
This tree is then signed by the developer, constituting release approval.
The signed trees are submitted to co-signing witness servers.

The witnesses require a threshold of valid developer signatures to accept a
package for release.
Additionally, the mapping between source and binary is verified by some of
the witnesses.
If these two checks succeed, the release is accepted and collectively signed
by the witnesses.

%
The system allows to rotate developer keys and witness keys, while the root of
trust is an offline key.
It also functions as a timestamping service, allowing for verification of
update timeliness.

Additionally to individual packages, releases of multiple packages are also
supported.
These ``snapshots'' are created by aggregating individual package skipchains.
Over the most recent versions of these, a Merkle tree is constructed and
signed by the witnesses.

\vspace{-3mm}
\paragraph{Transparency systems.}
Certificate transparency (CT) uses Merkle tree logs to provide a public view on
all certificates used in HTTPS~\cite{rfc6962, laurie_certificate_2014,
laurie_computing:_2012, ben_laurie_certificate_2011, rfc6962bis}.
Site operators can observe certificates issued for their domains and help
detect misissuances.
CT is widely deployed and its use continues to expand.
Basin et al.\ develop ARPKI, an alternative public key infrastructure for
server certificates~\cite{basin_arpki:_2014}.
Domain owners choose two certification authorities and an integrity log
server.
The system assumes that a number of public integrity log servers exist.
These are used to provide a globally consistent view on all the certificates
in existence.
Fahl et al.\ suggest a transparency model for Android applications, providing
``Application Transparency''~\cite{fahl_hey_2014}.
In order to secure application updates, the hashes of these are submitted
into a transparency log system.
Melara et al.\ apply Merkle tree-based auditing to secure mobile messaging in
the Continuous identity and key management system
(CONIKS)~\cite{melara_coniks:_2015}.  
The central premise of the approach is that users are changing their keys
frequently, and this process must provide a smooth and secure user experience.
Keys are bound to identities by submitting them into the audit log, which also
serves as a public key directory.

\section{Design}
\label{sse:design}

We propose to extend the signature based APT by a transparency system, adding
on top of the existing security mechanisms.
This system is modeled after the transparency overlay, where a log server
maintains a Merkle tree over a list of submitted items.
The honest operation of the log is verified by auditors and monitors.
In our case, the auditor is integrated into the APT client and the monitor is
a new standalone component.

\subsection{Threat model}
Additionally to existing Debian mechanisms, 
we introduce a log server that maintains a
Merkle tree as part of a transparency overlay.
The system consists of the following components.
Individual maintainers upload signed source packages to the archive.
The archive compiles and distributes binary packages, submitting releases into
the log.

Compromise of components and collusion 
of participants
must not result in a violation of the
following security goals remaining undetected.
A goal of our system is to make it infeasible for the attacker to deliver targeted backdoors.
For every binary, the system can produce the corresponding source code and the
authorizing maintainer.
Defined irregularities, such as a failure to correctly increment version
numbers, also can be detected by the system.

%
%



\subsection{Release log}

The APT release file identifies, by cryptographic hash, the packages, sources,
and meta data which includes dependencies.
This release file, meta data, and source packages are submitted to a log
server operating an append-only Merkle tree, as shown in
Figure~\ref{fig:overview}.
The log adds a new leaf for each file.

\begin{figure}[t]
    \centering
    \begin{tikzpicture}[
            execute at begin node=\strut,
            obj/.style={
                draw,
                rounded corners,
                minimum width=45mm,
            },
            log/.style={
                draw,
                circle,
            },
            submit/.style={
                dashed,
                ->,
            },
        ]
        \node[name=release, obj] { archive-signed release file };
        \node[name=meta, below=2mm of release, obj] { meta data };
        \node[name=src, below=2mm of meta, obj]
            { maintainer-signed source packages };
        \node[name=bin, below=of src, obj] { binary packages };

        \node[name=log, right=15mm of meta, log] { log };


        \draw[submit] (release.east) to node[above, sloped] {submit} (log); 
        \draw[submit] (meta.east) to (log); 
        \draw[submit] (src.east) to (log); 

        \draw[->] (src) to node[right] {reproducible} (bin);
        
    \end{tikzpicture}
    \caption{System overview.}
    \label{fig:overview}
    \vspace{-2mm}
\end{figure}
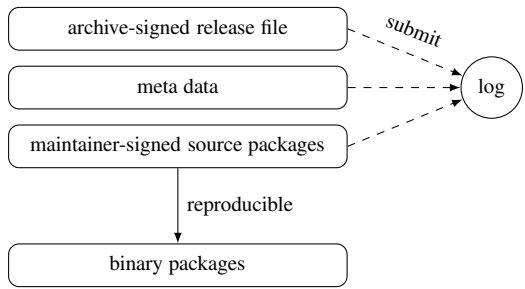

We assume maintainers may only upload signed source packages to the archive,
not binary packages.
The archive submits source packages to one or more log servers.
We further assume that the buildinfo files capturing the build environment are
signed and are made public, e.g. by them being covered by the release file,
together with other meta data.

In order to make the maintainers uploading a package accountable, a source
package containing all maintainer keys is created and submitted into the
archive.
This constitutes the declaration by the archive, that these keys were
authorized to upload for this release.
The key ring is required to be append-only, where keys are marked with an
expiry date instead of being removed.
This allows verification of source packages submitted long ago, using the keys
valid at the respective point in time.

At release time, meta data and release file are submitted into the log as
well.
The log server produces a commitment for each submission, which constitutes
its promise to include the submitted item into a future version of the tree.
The log only accepts authenticated submissions from the archive.
The commitment includes a timestamp, hash of the release file, log identifier
and the log's signature over these.
The archive should then verify that the log has produced a signed tree root
that resolves the commitment.
To complete the release, the archive publishes the commitments together with
the updates.
Clients can then proceed with the verification of the release file.

The log regularly produces signed Merkle tree roots after receiving a valid
inclusion request.
The signed tree root produced by the log includes the Merkle tree hash, tree
size, timestamp, log identifier, and the log's signature.

On the client side, the release file will be retrieved as usual.
Given the release file and inclusion commitment, the client can verify by
hashing that the commitment belongs to this release file and also verify the
signature.
The client can now query the log, demanding a current tree root and an
inclusion proof for this release file.
Per standard Merkle tree proofs~\cite{chase_transparency_2016,
crosby_efficient_2009, rfc6962}, the inclusion proof consists of a list of
hashes to recompute the received root hash. 
For the received tree root, a consistency proof is demanded to a previous known
tree root.
The consistency proof is again a list of hashes.
For the two given tree roots, it shows that the log only added items between
them.
Clients store the signed tree root for the largest tree they have seen, to be
used in any later consistency proofs.
Set aside split view attacks, which will be discussed later, clients verifying
the log inclusion of the release file will detect targeted modifications of
the release.

The procedures described do not add any tasks for an ordinary user of APT, or
for a maintainer.
One further consideration is the availability of the log.
To improve log availability, elements can be submitted into multiple
logs.
Clients would then contact all of these to validate a release, and require a
quorum.

\subsection{Removal of elements}

Source packages sometimes must be removed from the archive.
The reason is usually that the distribution license for a particular file in
the package is not acceptable under the project's guidelines.
The offending packages are then removed from the archive, and a new version of
the package is prepared, fixing the issues.

Since the log must be append-only, and also provide the source code for all
packages, this legal requirement cannot be fulfilled directly.
In order to perform the removal, a removal notice is submitted
to the log.
It consists of a statement signed by the archive, specifying which source
package was removed, the time and reason of removal.
It will be returned in response to requests for the original source.

\subsection{Hidden versions}

The hidden version attack attempts to hide a targeted backdoor by following
correct signing and log submission procedures.
It may require collusion by the archive and an authorized maintainer.
The attacker prepares targeted malicious update to a package, say
version~v1.2.1, and a clean update~v1.3.0.
The archive presents the malicious package only to the victim when it wishes
to update.
The clean version~v.1.3.0 will be presented to everybody immediately
afterwards.

A non-targeted user is unlikely to ever observe the backdoored version,
thereby drawing a minimal amount of attention to it.
The attack however leaves an audit trail in the log, so the update itself can
be detected by auditing.

A package maintainer monitoring uploads for their packages using the log would
notice an additional version being published.
A malicious package maintainer would however not alert the public when this
happens.
This could be construed as a targeted backdoor in violation of the stated
security goals.

To mitigate this problem a minimum time between package updates can be
introduced.
This can be achieved by a fixing the issuance of release files and their log
submission to a static frequency, or by alerting on quick subsequent updates to
one package.

In the hidden version attack, the attacker increases a version number in order
to get the victim to update a package.
The victim will install this backdoored update.
The monitor detects the hidden version attack due to the irregular release file
publication.
There are now two cases to be considered.
The backdoor may be in the binary package, or it may be in the source package.

The first case will be detected by monitors verifying the reproducible builds
property.
A monitor can rebuild all changed source packages on every update and check if
the resulting binary matches.
If not, the blame falls clearly on the archive, because the source does not
correspond to the binary, which can be demonstrated by exploiting reproducible
builds.

The second case requires investigation of the packages modified by the update.
The source code modifications can be investigated for the changed packages,
because all source code is logged.
The fact that source code can be analyzed, and no analysis on binaries is
required, makes the investigation of the hidden version alert simpler.
The blame for this case falls on the maintainer, who can be identified by
their signature on the source package.
If the upload was signed by a key not in the allowed set, the blame falls on
the archive for failing to authorize correctly.

If the package version numbers in the meta data are inconsistent, this
constitutes a misbehavior by the submitting archive.
It can easily be detected by a monitor.
Using the release file the monitor can also easily ensure, by demanding
inclusion proofs, that all required files have been logged.



%
%
%
%

\section{Log validation}
\label{sse:monitor}

Monitors are a required component to ensure the log operates correctly as well
as a building block against equivocation.
Monitors communicate with the log, verifying consistency.
In case of misbehavior, a monitor would raise an alert and provide the
cryptographic proof that comes with it.
Many monitors can observe one log.

Monitors retrieve all additions into the log, allowing them to execute custom
investigation functions on the logged items.
The consistency of the log is always checked before the items are processed
for flagging.
Each monitor can have different investigation rules, or none at all, in which
case it only monitors the append-only operation of the log.

\subsection{Log monitoring functions}

The primary function of a monitor is to ensure that the log server maintains
its list of items by only appending items.
The monitor initializes by retrieving all items from the log server, as well
as the signed tree root.
The monitor verifies the correctness of the tree root by recomputing the hash
tree and verifying the signature.

It will now continuously poll the log server if a new tree root is available.
Should a new tree root be published, the monitor retrieves the tree root and
all items that were added since the previous tree root the monitor had
observed.
The monitor can now recompute the new hash tree in order to verify the tree
root signed by the log.

In this approach, the monitor stores all the items in the list maintained by
the log.
It will also observe any new items.
It is therefore possible for the log to investigate all items and raise alerts
based on the findings.

%
%


\subsection{Monitor examination of events}

Additionally to their fundamental task in keeping the log honest, the monitors
can fulfill additional roles by investigating the items kept in the log.
In the following, several such functions are described.

For each of these, we will assume that the monitor continuously updates its
view of the log by retrieving any new items the log has added.
Before calling the investigating functions, the list of new items is filtered.
In general, we are only interested in release files.
The signature on the release file is verified.
The functions are then called for each of the items remaining after filtering
and verification.
We will also assume that the investigating functions have access to the log
content.
They need to be able to determine if an element is part of the log,
identify the preceding release file and parse it.

\subsubsection{Release file consistency}
A basic check that requires a monitor is to verify the consistency of the
release file.
This check ensures that the archive submitted a consistent meta data
state into the log.

\begin{algorithm}[ht]
\vspace{-3mm}
\begin{codebox}
    \Procname{Completeness(releaseFile):}
    \li \If \Not verifySignature(releaseFile):
        \li \Then \Alert releaseFile
    \End
    \zi \Comment indices (Packages, Sources) logged
	\li \For indexFile \In releaseFile: \Do
		\li \If \Not isInLog(indexFile):
            \li \Then \Alert indexFile
        \End
    \End
    \zi \Comment source packages logged
    \li \For sourcesFile \In releaseFile.sourcesFiles: \Do
        \li \For sourcePkg \In sourcesFile: \Do
            \li \If \Not isInLog(sourcePkg):
                \li \Then \Alert sourcePkg
            \End
        \End
	\End
    \li \Return
\end{codebox}
    \vspace{-5mm}
    \caption{Verify elements covered by a release file are logged.}
    \label{alg:coverage}
\end{algorithm}

The first process is Algorithm~\ref{alg:coverage}.
The following verifications are executed for each new release
file added to the log.
For each \textit{Sources} and \textit{Packages} file, their presence in the
log, with matching hash, must be ensured.
For every Sources file, all listed source packages must be in the log with a
matching hash.

\begin{algorithm}[ht]
\begin{codebox}
    \Procname{SourceAvailable(releaseFile):}
    \li sourcesFiles \gets releaseFile.sourcesFiles
	\li \For packagesFile \In releaseFile.packagesFiles: \Do
        \li \For package \In PackagesFile: \Do
            \zi \Comment no source available
            \li \If \Not sourcePresent(package, sourcesFiles):
                \li \Then \Alert package, sourcesFiles
            \End
        \End
    \End
    \li \Return
\end{codebox}
    \vspace{-5mm}
    \caption{Verify that all sources are available.}
    \label{alg:src}
\end{algorithm}

Each binary package, as enforced by Algorithm~\ref{alg:src}, must come
with a corresponding source package.
This is necessary in itself and also for later checks such as
reproducibility.

\begin{algorithm}[ht]
\vspace{-2mm}
\begin{codebox}
    \Procname{VersionConsistency(releaseFile):}
    \li sourcesFiles \gets releaseFile.sourcesFiles
	\li \For packagesFile \In releaseFile.packagesFiles: \Do
        \li \For package \In PackagesFile: \Do
            \zi \Comment compare with previous release file
            \li \If \Not metaChanged(package):
                \li \Then \Continue
            \End
            \li \If \Not versionIncremented(package):
                \li \Then \Alert package
            \End
            \li source \gets getSource(package, sourcesFiles)
            \li \If \Not metaChanged(source) \And
                \li \quad \Not buildinfoChanged(package):
                \li \Then \Alert package, source
            \End
        \End
    \End

	\li \For sourcesFile \In sourcesFiles: \Do
        \li \For source \In sourcesFile: \Do
            \li \If \Not metaChanged(source):
                \li \Then \Continue
            \End
            \li \If \Not versionIncremented(source):
                \li \Then \Alert source
            \End
        \End
    \End

    \li \Return
\end{codebox}
    \vspace{-5mm}
    \caption{Verify version consistency of the release file.}
    \label{alg:version}
\end{algorithm}

We enforce version bumps for modified packages by comparing the meta data
fields of all source and binary packages.
The meta data contains in particular the hash, version number, and
dependencies.
If this is not enforced, a client might see and install a new package with
different meta data than others.

For this purpose, the monitor maintains several data items associated with the
package name.
Out of the Packages and Sources files, the package versions and meta data
block are extracted.
The meta data contains in particular the package hash and its dependencies.


For the given package names, the entire entry with meta data and hash of
the package is compared for equality to the stored one.
In case of changes, the local package entry is updated and the configured
notification actions taken.
If a change in meta data is detected in Algorithm~\ref{alg:version}, the version
number must be incremented as well.



Monitors may provide analysis functions additional to those discussed so far.
In such a case, the monitor should meet the checks of Algorithms
\ref{alg:coverage}, \ref{alg:src}, and \ref{alg:version}.
These ensure that the meta data covered and provided by the release file is in
a sane state and fit for further analysis.

\subsubsection{Frequency}

Another monitor function is to observe the frequency in which release files
are produced and made available.
The monitor needs to closely follow the log and know the expected release
schedules.
The monitor will produce an alert when an irregular release interval occurs.

This mechanism enables investigators to focus attention on possible
misbehavior.
There should also be an alert, if the archive stops publishing releases
unexpectedly.
This process is asynchronous and runs continuously.

If an alert occurs, the monitor also has the necessary tools to help
investigation.
As it has all source packages available, it can produce the difference in
source code compared to the last release.
The signatures on the source packages identify the maintainers involved.
If the signature was produced by a valid key, the maintainer is responsible
for any issues the source code exhibits.
In case the signature is invalid, the archive has allowed violation of the
submission access control.

\subsubsection{Package maintainers}

To verify the adherence of the archive to the package upload ACL, all source
packages that were changed in a release file are investigated.
A package is also investigated if its meta data or buildinfo files changed.

For the updated packages, if any of the signatures cannot be verified as
having been created by a valid maintainer key, an alarm is raised.
Not all uploaders may update every package.
In the policy map of signing keys acceptable for each package, keys of full
members are added to every package, other uploader keys are added to their
respective packages.


Additionally to this generic check, individual maintainers of packages might
want to see which of their packages were updated, and by whom.
After submission, maintainers keep track if their package was published.
They should also observe if a new upload signing key is published under their
name.
The maintainers of the upload ACL should also keep track of the keys published
by the archive.



%
%
%



\subsubsection{Reproducible builds}

For every package that changes, a monitor can verify the reproducible build
property.
This check is important to make sure that the archive publishes the software
intended by the maintainer, shown in Algorithm~\ref{alg:rebu}.
If the relationship between source and binary is not verified, the archive
could build backdoors into binaries, where the backdoors are not reflected in
the source code.

\begin{algorithm}[ht]
\vspace{-2mm}
\begin{codebox}
    \Procname{Reproducible(releaseFile):}
    \li \For package \In releaseFile.allpackages: \Do
        \li \If metaChanged(package):
            \Then \zi \Comment rebuild for all architectures and compare
            \li correct \gets isRepro(package, releaseFile)
            \li \If \Not correct:
                \li \Then \Alert package
            \End
        \End
    \End
    \li \Return
\end{codebox}
    \vspace{-5mm}
    \caption{Verify transformation from source to binary.}
    \label{alg:rebu}
\end{algorithm}

Given a new release file, the monitor would determine all modified source
packages by comparing the new Sources file to the previous one.
All changed packages are now rebuilt.
The hashes of the binary packages must match the hashes provided in the
Packages files.
If any mismatches occur, an alarm is raised.
The blame falls on the archive, which has published a non-reproducible
package, possibly including a backdoor.

\section{Equivocation}
\label{sse:equivocation}

The most significant attack by the log or with the collusion of the log is
equivocation.
In a split-view or equivocation attack, a malicious log presents different
versions of the Merkle tree to the victim and to everybody else.
Each tree version is kept consistent in itself.
The tree presented to the victim will include a leaf that is malicious in some
way, such as an update with a backdoor.
It might also omit a leaf in order to hide an update.
This is a powerful attack within the threat model that violates the security
goals and must therefore be defended.
A defense against this attack requires the client to learn if they are served
from the same tree as the others.

In some log systems equivocation is addressed with gossiping, or with
monitors and auditors~\cite{chase_transparency_2016, chuat_efficient_2015,
crosby_efficient_2009, nordberg_gossiping_2017}.
The auditor, as a functionality embedded with the client, exchanges tree roots
with the monitor.
They request consistency proofs between their own observed tree root and the
tree root of the other party.
Both are then able to notice if the log presents a different view to the other
party, detecting equivocation.
Because the tree roots are signed, the log can be attributed as the malicious
party.

To some extent, we envision that monitors are run much like package mirrors are
today.
Some organizations may donate their resources for the public good, others may
opt to just improve their own operations.
To provide clients with a monitor after a new installation, a list of
trustworthy monitors would need to exist.
For the distribution to be confident in their reliability, it ultimately would
need to run them.
This runs contrary to the monitor's task of holding the distribution
accountable.
It begs the question how a meaningful trust boundary can be drawn through the
distribution infrastructure.
We conclude that auditor-monitor communication is not sufficient to address
equivocation in our case.

Log operations can be done on multiple logs e.g., by submitting everything into
two logs.
While this does help with equivocation in principle, the logs are in our
scenario operated by the same group.
It is therefore prudent to assume that if one of these is compromised, the
second one would also be.

This approach can be used for redundancy, where a quorum of logs are queried
and must provide the expected answer.
If one log fails or must be taken offline after misbehaving, the clients could
still be satisfied with responses of the other logs.

\subsection{Design for tree root cross logging}

There is limited advantage to using multiple logs by the same operator.
Defenses against equivocation are closely related to operator diversity.
To extend the pool of eligible secondary log operators, interoperation between
logs of different organizations or even domains is required.
To make this interoperation realistic, the interface should be small and
simple, which we will strive for in the following.



One approach for detecting equivocation is logging of tree roots between
cooperating logs.
This approach requires multiple log servers.
We now assume different logs exist, run by different operators.
These could be other Linux distributions, or possibly a Certificate
Transparency log made compatible.

Consider one log accepting new list items, for now dubbed the committing log.
This log regularly creates a new signed tree root on inclusion of list items.
This log will now submit this tree root, when it is created in response to a
new release file, into another log, the witnessing log.
Additionally to performing its native tasks, the witnessing log accepts the
submitted tree root as a new list element and includes it into its Merkle tree.
The inclusion commitment retained from the witnessing log, as well as the
submitted tree root, will be provided to the archive by the committing log.

\begin{figure}[t]
    \centering
    \begin{tikzpicture}[
            scale=.75,
            level distance=5mm,
            level/.style={
                sibling distance=6em/#1,
            },
            edge from parent/.style={
                draw,
            },
            inner/.style={
                draw,
                circle,
                minimum width=2.2mm,
                inner sep=0,
            },
            leaf/.style={
                draw,
                rectangle,
                minimum width=1.8mm,
                minimum height=1.8mm,
                inner sep=0,
            },
            marked/.style={
                draw,
                fill=gray!50,
                circle,
            },
            cross/.style={
                shorten >=\pgflinewidth,
                shorten <=\pgflinewidth,
            },
        ]
        \node[name=roota, inner] {}
        child {
            node[name=l, inner] {}
                child {
                    node[name=n0, inner] {}
                }
                child {
                    node[name=n1, inner] {}
                }
        }
        child[missing]
        ;

		\node[left=of roota] { 1.~a) };
		\node[name=n2, right=2em of n1, inner] {};
		\draw (roota) to (n2);

        \foreach \n in {0, ..., 2} {
            \node[name=l\n, below=1em of n\n, leaf] {};
            \draw (l\n) to (n\n);
        }   
    
        \node[name=rootb, right=30mm of roota, inner] {}
        child {
            node[inner] {}
                child {
                    node[name=p0, inner] {}
                }
                child {
                    node[name=p1, inner] {}
                }
        }
        child {
            node[inner] {}
                child {
                    node[name=p2, inner] {}
                }
                child {
                    node[name=p3, inner] {}
                }
        }
        ;

        \foreach \n in {0, ..., 3} {
            \node[name=m\n, below=2mm of p\n, leaf] {};
            \draw (m\n) to (p\n);
        }   
    
		\node[left=of rootb] { 1.~b) };
		\node[marked, minimum width=2.2mm, inner sep=0,] at (rootb) {};
		\node[marked, leaf] at (m3) {};

		\begin{scope}[
				level/.style={
					sibling distance=10em/#1,
				},
		]
        \node[name=root2, inner,
				below=12mm of p0,
                xshift=-3mm,
			] {}
            child foreach \x in {0,1} {
                node[name=q\x, inner] {}
                child foreach \y in {0,1} {
                    node[name=q\x\y, inner] {}
                    child foreach \z in {0,1} {
                        node[name=q\x\y\z, inner] { }
                    }
                }
            }
        ;
		\end{scope}
        
        \foreach \x in {0, 1}
            \foreach \y in {0, 1}
                \foreach \z in {0, 1} {
                    \node[name=i\x\y\z, below=2mm of q\x\y\z, leaf] {};
                    \draw (i\x\y\z) to (q\x\y\z);
                }
        ;

		\node[left=29.5mm of root2] { 2. };
		\node[marked, leaf] at (i111) { };
    
		\draw[dashed, ->] (rootb) to [out=0, in=45, out looseness=1.75] (i111);
    \end{tikzpicture}
    \caption{Tree root cross logging.}
    \label{fig:cross}
    \vspace{-3mm}
\end{figure}
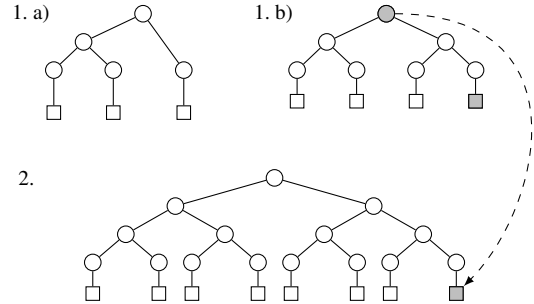

This relationship is shown in Figure~\ref{fig:cross}.
The upper row shows one log of list size three adding a fourth element.
The lower row shows a different log.
The log with three elements in situation 1a adds an element in step 1b,
resulting in a new tree root.
It then submits this tree root in step 2 as element for log inclusion to the
lower log.
The lower log accepts this new element into its list.
By publishing the tree root into the witnessing log, the primary log publicly
commits to a history of its tree.
This commitment can be used to detect equivocation through monitoring of the
witnessing log.

%
%
%


When the client now verifies a log entry with the committing log, it also has
to verify that a tree root covering this entry was submitted into the
witnessing log.
Additionally, the client verifies the append-only property of the witnessing
log.

The witnessing log introduces additional monitoring requirements.
Next to the usual monitoring of the append-only operation, we need to check
that no equivocating tree roots are included.
To this end, a monitor follows all new log entries of the witnessing log that
are tree roots of the committing log.
The monitor verifies that they are all valid extensions of the committing
log's tree history.

By using tree roots as a generic interoperability layer between logs,
cooperation between different operators and domains is enabled.
Each log only needs to add one new leaf type in order to participate.
This is easier to achieve than, for example, making a CT log compliant with all
the requirements for software transparency logging.

\subsection{Auditor}

To contribute to the security provided by root cross logging, an auditor
component needs to contact the second log.
The auditor verifies that its own view on the log is committed to the
witnessing log.
It validates the correct operation of this second log with the established
proof mechanisms.

The committing log, denoted as log~A here, has submitted a tree root into the
witnessing log, log~B.
The auditor is now assumed to be provisioned by the archive with this tree
root and the corresponding inclusion promise from log~B.
An element, in our case the release file, is verified by requesting an
inclusion proof for it from log~A.
The inclusion proof is requested specifically for the tree size of the root
given by the archive, ensuring that it actually covers this element.
The auditor then requests a consistency proof from that tree root to a previous
known tree root.
After verifying the two proofs, we have now established that log~A has
submitted a tree root into log~B and that this tree root is consistent with the
auditor's view on the tree.


We proceed to verify inclusion into the witnessing log.
Using the knowledge of the committed tree root and the inclusion promise by
the witnessing log, we can request an inclusion proof from the witnessing log.
For the tree root returned with that inclusion proof and a previous known tree
root, we request a consistency proof from the witnessing log.




The auditor has now established that log~A has presented the same log view to
the auditor and to log~B.
Crucially, the auditor can not be certain that its view is the only one that
log~A has committed into log~B.
To ensure that, every element in log~B must be inspected.
This task naturally falls to a monitor, which is discussed next.

\subsection{Monitor}

The task of the monitor in cross logging is to make sure that one log only
commits to one view.
This requires the monitor to keep track of both the committing and the
witnessing log.
It monitors the committing log, log~A.
For simplicity, we assume that the monitoring process consists of downloading
and storing all new log elements, and verifying signed tree roots.
The knowledge of log~A fixes the monitor's view on the log and enables it to
generate tree roots for all tree sizes.
This is the usual procedure for monitoring.

Extending this, the monitor will also monitor the witnessing log~B.
In the following, all list elements of log~B are investigated.
We are interested in all elements that are tree roots of log~A, so we filter
for these.

For each of these elements, first the signature is verified.
In a second step, the Merkle tree root in this element is recomputed using the
knowledge of the elements of log~A.
The computed tree root is compared to the element in log~B.
If these checks succeed, the element corresponds to the view on log~A that the
monitor has.

If the verification procedure succeeds for all elements that are signed tree
roots of log~A, then the monitor has established that log~A has only ever
committed one view to its witness log~B.
Should a tree root not correspond to the known list of elements of log~A,
equivocation has been detected.


If log~B is dishonest, it cannot frame log~A, because tree roots are signed.
Log~B also cannot omit elements.
Log~A has obtained inclusion promises from log~B, enabling it to demonstrate
that a given tree root was submitted into log~B.
If both logs collude in equivocation, no security can be achieved in this
system.

\subsection{Security}

The approach presented in this paper protects the clients directly that
additionally check the witnessing log, if the logs do not collude.
If the committing and witnessing logs do not collude, the committing log will
be unable to observe that the client verifies its honesty by querying the
witnessing log.
It is therefore unable to distinguish between clients that do this additional 
check and those that do not.
This results in herd immunity for clients that do not check the witnessing log,
provided that a proportion of clients actually does this check.
In any case, it is necessary that some of the clients of the committing log
verify the witnessing log.


The more logs participate by including the tree roots of the other logs, the
harder attacks become.
One honest log is enough to detect equivocation.
For package authentication, each distribution could run one or two logs, all
of which cooperate by tree root logging.
To mitigate the risk of collusion, multiple logs run by different operators
are required.


\section{Evaluation}
\label{sse:evaluation}

In the following, we will informally analyze the security properties of the
main design.
We further implement the system and feed it two years worth of distribution
updates, noting performance characteristics and detected irregularities.

\subsection{Software transparency as a secure pledged transparency overlay}

We claim that our log design described previously implements a
secure transparency overlay as proposed and proven by Chase and
Meiklejohn~\cite{chase_transparency_2016}.
This constitutes our primary security argument, extended in separate arguments
by monitoring functions and protection against equivocation.

In the following we show how a secure overlay is instantiated from the
proposed software transparency mechanism.
In the parlance of the secure overlay, the log, monitor, and auditor are part
of the ``overlay''.
The ``system'' is an existing infrastructure to be secured with the
transparency overlay.
In our case, the system is the software distribution via the archive and the
APT client.

\newtheorem{theorem}{Theorem}
\begin{theorem}
    Software transparency is a secure pledged transparency overlay.
    \label{theo}
\end{theorem}

\begin{proof}
    We use the method of Chase and Meiklejohn to instantiate the secure pledged
    transparency overlay.
    Two parts need to be demonstrated.
    First, the function creating the system events (\texttt{GenEventSet})  needs to
    be instantiated.
    Second, we need to define for all transparency overlay protocols which parts of
    the software transparency system interact with them, namely for the
    \texttt{Log} and \texttt{CheckEntry} protocols.
    The \texttt{Log} protocol defines submission of events into the log.
    In the \texttt{CheckEntry} protocol, the system interacts with the auditor
    component of the transparency overlay.
    The auditor in turn verifies log inclusion and log consistency, given an event
    and a corresponding inclusion promise of the log.

    \texttt{GenEventSet.} \quad
    Elements are solely generated by the archive.
    They can be source package files, meta data files, and signed release
    files.

    \texttt{Log protocol.} \quad
    In the \texttt{Log} protocol, elements are submitted to the log.
    The archive is the sole originator of events and submits these to the log.

    \texttt{CheckEntry protocol.} \quad
    We designate this part of the system to be the APT client, making the auditor
    effectively part of the APT client.
\end{proof}

We note that a release is now produced jointly by the archive and the log, not
just the archive anymore.
The client is now able to verify the transparency property on release files.

The functions of the log, the auditor and monitor are implemented as
prescribed by the overlay.
We conclude that the proposed software transparency mechanism constitutes a
secure transparency overlay.
This allows the system to build on proven security properties.




\paragraph{Detection of targeted backdoors.}

Using the secure overlay, we will demonstrate that our system detects targeted
backdoors, achieved by tailored monitor services.
Assume a malicious or compromised maintainer inserts a backdoor into a source
package.
In order to be accepted into the archive, it must be signed by their key.
Should the archive not enforce this, a monitor verifying maintainer signatures
would notice and alert.
It can therefore be attributed to the maintainer and the source code is
present for analysis.
A maintainer cannot upload binary packages without help by the archive,
because we assume only source uploads are allowed.



If the archive modifies a source package, this package lacks a valid
maintainer signature.
This is detected by a monitor validating the maintainer signatures and upload
policy.
If the archive modifies a binary package without modifying the corresponding
source, this is detected by a monitor verifying the reproducible builds
property.
Modifications by the log are discovered, because the release file signed by
the archive covers all other files.

\subsection{Performance}

Our system uses the Trillian generic Merkle tree
implementation~\cite{_trillian:_2017} which is also employed in some CT logs.
We implement our log functionality using this hash tree.
The log offers an HTTPS interface for JSON objects for element submission and
the auditor and monitor functionality.
The log is instantiated with the SHA256 hash function and the ECDSA-P256
signature scheme.

Starting at the inception of the Debian ``stretch'' release, we replay two
years of Debian updates.
The historic updates are retrieved from the \url{snapshot.debian.org} service
before the experiment.
This results in 3010 release files, starting on 2015-04-25 and ending on
2017-06-17.
The ``stretch'' release is, at that point in time, a rolling release.
As such it experiences much more updates than a final release that only gets
updates for security support.

In our measurement setup, the log resides on a different machine than the
submission component, auditor, and monitor.
The network communication with the log consists of HTTPS requests.
For each of the release files, we first submit the release meta data and
source packages into the log and then run the auditor afterwards.
Given the inclusion promise for the release file, the auditor validates
the log operation.
It retrieves an inclusion and a consistency proof with GET requests from the
log and verifies these.
After all releases are logged and validated, the monitor fetches the log
elements using GET requests and executes its verification functions.

\begin{figure}[t]
    \centering
    \begin{tikzpicture}
        \begin{axis}[
                height=40mm,
                width=.9\columnwidth,
                boxplot/draw direction=y,
                xtick={1,2,3,4},
                xticklabels={ incl. tx, incl. rx, cons. tx, cons. rx },
                ylabel=traffic per request,
                y unit=kB,
                ytick={0, 1000, 2000, 3000},
                yticklabels={0, 1, 2, 3}, 
                ymin=0,
            ]
            \addplot[boxplot] table [y=incl_tx, col sep=comma] {audit.dat};
            \addplot[boxplot] table [y=incl_rx, col sep=comma] {audit.dat};
            \addplot[boxplot] table [y=cons_tx, col sep=comma] {audit.dat};
            \addplot[boxplot] table [y=cons_rx, col sep=comma] {audit.dat};
            
        \end{axis}
    \end{tikzpicture}
    \vspace{-5mm}
    \caption{Auditor traffic to the log.}
    \label{fig:auditor}
    \vspace{-2mm}
\end{figure}
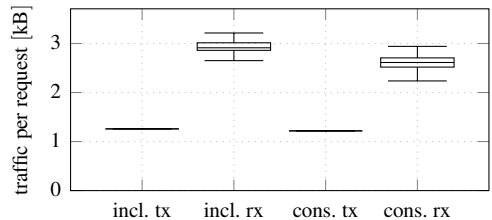

We observe that the duration of the log submissions is dominated, by several
orders of magnitude, the transfer duration incurred by submitting source
packages.
As the auditor runs on every client, its performance is most critical.
In terms of computation, only some signature verifications and hash
computations are added.
Next, we measure the network layer PDU traffic between auditor and log, using
the Linux \texttt{net\_cls} cgroups and netfilter.
The data in Figure~\ref{fig:auditor} shows the traffic caused by inclusion
proofs and consistency proofs during auditor operation.
The request for the inclusion proof requires sending 1.3\,kB and
receiving 2.9\,kB on average.
The consistency proof is requested for the tree size of the previous release
file.
The consistency proof requires receiving slightly less data, on average 2.6\,kB.
Note that these numbers include the TLS handshake.


One optimization that can be done for the client is to deliver the commonly
required proofs over the mirror instead of contacting the log separately.
There need to be several generations of consistency proofs on the mirror, in
order to support clients that do not retrieve every new release.
For 28~generations of release files covering one week, the consistency proofs
require about 18\,kB.
This method scales linearly in storage size per generation.

\begin{figure}[t]
    \centering

    \begin{subfigure}{.49\linewidth}
        \begin{tikzpicture}
            \begin{axis}[
                    width=\textwidth,
                    height=30mm,
                    xlabel=number of leaves,
                    xtick={100000, 200000},
                    xticklabels={1, 2},
					every x tick scale label/.style={
						at={(xticklabel* cs:1.03,0cm)},
                        anchor=near xticklabel,
                    },
                    ylabel=size,
                    y unit=MB,
                    ytick={0, 100000000, 200000000, 300000000, 400000000},
                    yticklabels={0, 100, 200, 300, 400},
                    ymin=0,
                    scaled y ticks=false,  
                ]
                \addplot [] table
                    [x=tree_size, y=db_du, col sep=comma]
                    {submit.dat};
            \end{axis}
        \end{tikzpicture}
        \caption{Tree database in MB.\strut}
        \label{fig:du-log}
    \end{subfigure}
    \begin{subfigure}{.49\linewidth}
        \begin{tikzpicture}
            \begin{axis}[
                    width=\textwidth,
                    height=30mm,
                    xlabel=number of leaves,
                    xtick={100000, 200000},
                    xticklabels={1, 2},
					every x tick scale label/.style={
						at={(xticklabel* cs:1.03,0cm)},
                        anchor=near xticklabel,
                    },
                    ylabel=size,
                    y unit=GB,
                    ytick={0, 100000000000,200000000000,300000000000,400000000000},
                    yticklabels={0, 100, 200, 300, 400},
                    ymin=0,
                    scaled y ticks=false,
                ]
            \addplot [] table
                [x=tree_size, y=fs_du, col sep=comma]
                {submit.dat};
            \end{axis}
        \end{tikzpicture}
        \caption{Source packages in GB.}
        \label{fig:du-disk}
    \end{subfigure}
    \vspace{-1.5mm}
    \caption{Disk usage of the log server.}
    \label{fig:log}
\end{figure}
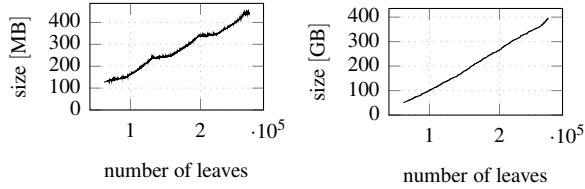

On the log side, the disk usage grows with the submitted elements.
The graphs in Figure~\ref{fig:log} show the occupied storage against the tree
size.
After two years, with around 270\,000 elements logged, the database uses
443\,MB and the package contents occupy an additional 396\,GB.

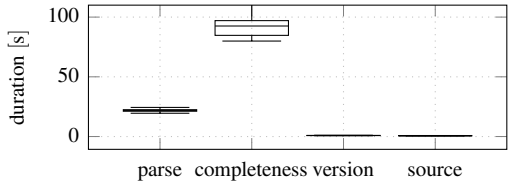
\begin{figure}[t]
    \centering
    \begin{tikzpicture}
        \begin{axis}[
                boxplot/draw direction=y,
                xtick={1,2,3,4},
                xticklabels={\strut parse,\strut completeness,\strut
                    version,\strut source},
                ylabel=duration,
                y unit=s,
                ymax=110,
                height=35mm,
                width=.9\columnwidth,
            ]
            \addplot[boxplot] table
                [y=parse, col sep=tab]
                {mon_t_func.dat};
            \addplot[boxplot] table
                [y=completeness, col sep=tab]
                {mon_t_func.dat};
            \addplot[boxplot] table
                [y=version, col sep=tab]
                {mon_t_func.dat};
            \addplot[boxplot] table
                [y=src, col sep=tab]
                {mon_t_func.dat};
        \end{axis}
    \end{tikzpicture}
    \vspace{-4mm}
    \caption{Processing time for monitor functions.}
    \label{fig:mon_t}
\end{figure}

The monitor operates asynchronously, but should issue alerts timely.
The processing time of the different monitor functions implemented is shown in
Figure~\ref{fig:mon_t}.
Parsing the release file and meta data takes about 20 seconds.
Checking if all files are logged takes over 80 seconds, as the process
accesses all files logged for the release and compares the cryptographic hash.
The version comparison and check for presence of a source package are then
negligible.

We conclude that the performance of all measured functions is sufficient for
actual use.
Even smaller projects should be capable of running a log server logging source
packages.

\subsection{Detected irregularities}

The task of the monitor is to flag suspicious elements in the log and raise
alerts.
We now discuss anomalies discovered by our monitor functions when applied to
the historic data.

\begin{figure}[t]
    \centering
    \begin{tikzpicture}
        \begin{axis}[
            height=40mm,
            width=.9\columnwidth,
            date coordinates in=x,
            xticklabel={\year-\month},
            ylabel=number of packages,
        ]
        \addplot [] table
            [x=date, y=src-miss, col sep=tab]
            {mon_f_alert.dat};
        \end{axis}
    \end{tikzpicture}
    \vspace{-2mm}
    \caption{Corresponding source missing.}
    \vspace{-3.5mm}
    \label{fig:src-miss}
\end{figure}
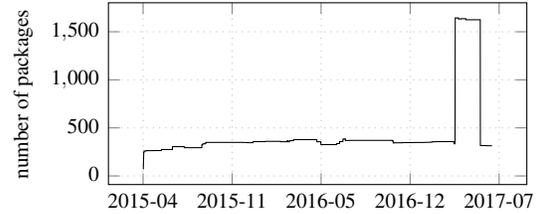

As shown in Figure~\ref{fig:src-miss}, there are continuously several
hundred binary packages for which there is no corresponding source package in
the release.
Note that this function counts packages multiple times when the source is not
identified for multiple architectures.
Starting March 2017, there is a large number of sources missing, indicating an
error which was fixed in May 2017.

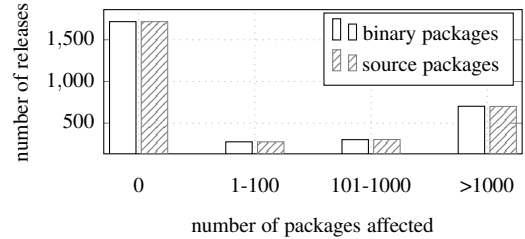
\begin{figure}[t]
    \centering
    \begin{tikzpicture}
        \begin{axis}
            [
                width=.9\columnwidth,
                height=35mm,
                ybar,
                xlabel=number of packages affected,
                symbolic x coords={0, 1-100, 101-1000, >1000},
                xtick style={draw=none},
                ylabel=number of releases,
            ]
            \addplot[fill=white] coordinates {(0, 1717)  (1-100, 277)  (101-1000, 303) (>1000, 703)};
            \addplot[color=gray, pattern=north east lines, pattern color=gray] coordinates {(0, 1717)  (1-100, 277)  (101-1000, 305) (>1000, 701)};
            \legend{binary packages, source packages}
        \end{axis}
        
    \end{tikzpicture}
    \vspace{-1mm}
    \caption{Version increment missing.}
    \label{fig:version}
    \vspace{-4mm}
\end{figure}

Comparison of version numbers was implemented using the Python
\texttt{apt\_pkg} library, which implements the non-trivial versioning rules
used in APT.
For the comparison of package meta data to detect if a package was changed we
filter non-critical fields such as tags.
Note that the meta data includes a hash over the package contents.
Whenever the meta data changes, we expect a version
increase.
All CPU architectures in the release are counted.
We show a many version increments are missing in a release in
Figure~\ref{fig:version}.
There are 1717 releases without unexpected events.
Some releases have a few inconsistencies in version increments, and a number of
releases have thousands.
The results for source and binary packages look similar in the plot, but
differ slightly.
There are a substantial number of changes in the package meta data without
version increment, suggesting that meta data changes are regularly done
without incrementing the version.

One source package that existed in the release meta data was missing in all
releases, because it had been removed administratively from the
\url{snapshot.debian.org} service.
There were no files where the hash was indicated incorrectly by the meta
data.

We discovered a considerable amount of anomalies in two years worth of
updates.
Our system would have discovered them automatically shortly after their
release and raised an alarm.
The result suggests that the release process must become more stringent in
order to allow such irregularities to be the cause of security alerts.

\section{Comparison to related work}
\label{sse:comparison}

CHAINIAC~\cite{nikitin_chainiac:_2017} is a software update transparency
system.
Clients can verify they were presented with a recent software version using
timestamping by the witnesses.
The witnesses also serve the function of assuring global consistency by
providing a collective signature.
The release history is fixed as immutable by inclusion of historic hashes.
In our proposed architecture, timeliness is assured by a wall clock validity
period embedded in the release file.
Global consistency and immutability is provided by submitting releases to the
log server.
The correct operation of the log is ensured by auditors and monitors.

The \emph{build verifiers} in CHAINIAC can be compared to \emph{monitors}
running our reproducible builds verification algorithm.
In contrast to the build verifier, the monitor operates asynchronously.
This means that the build verifier will detect and effectively stop
distribution of a non-reproducible package.
It also results in the system having to wait for enough build verifiers to
complete the build process for all supported instruction set architectures.
The monitor does not delay an update because a package builds slowly,
leading to a faster release compared to CHAINIAC.

In CHAINIAC, all code submissions must be confirmed by other developers
through cryptographic signatures.
The proposed architecture does not include a code review functionality, but
rather supposes a mapping of expected developers to projects, verified by
monitors.
In particular, the maintainers of the distribution's upload policy have an
incentive to monitor this activity and alert the wider project in case of
irregularities.

The collective authority consisting of independently operated witnesses
provides collective signatures, where some of the witnesses may be compromised.
The proposed architecture, on the other hand, relies on a single log or quorum
of logs.
This is justified by two reasons.
First, the logs require substantial disk space to store all versions of all
source packages ever submitted.
A log therefore need not only resources to operate the tree, but also
secondary storage, over time in the order of terabytes, which the
generic witness servers might be unwilling to provide.
Second, it is unclear who would operate the independent witness servers for
Linux distributions.
To achieve protection against equivocation, we advocate for the simple
interface provided by tree root cross logging between different logs.

The CHAINIAC model does provide accountability, but does not offer the
same forensic assurances.
In our proposed architecture, investigators are guaranteed an audit trail of
who made changes and the changes in \emph{source form}.
This guarantee is provided by the log, the operations of which
are verified by monitors.

In conclusion, our proposed system follows an ``untrusted but trustworthy''
model, doing validity checks only after a release.
This allows the system to reflect operational reality by exploiting the
incentives of different roles in the project.
Protection against equivocation is achieved by allowing independent
organizations to interoperate easily to their mutual benefit.

Parts of the system were proposed previously by the authors~\cite{blinded}.
In comparison, it is now extended with auditability, protection against
equivocation, and evaluated on actual distribution updates.

\section{Conclusion}
\label{sse:conclusion}
%
%

We propose a software transparency system for the popular Linux package
manager APT.
The system detects targeted backdoors and offers several forensic guarantees,
such as the availability of auditable source code for every installed binary.

To detect equivocation in log systems, we propose tree root cross logging.
This allows logs of different operators and types to interoperate, requiring
only a small interface.

An evaluation of the system on two years worth of actual distribution updates
identifies numerous anomalies, for instance missing source code.
Our system would have identified and allowed fixing these issues.


By tracking code and maintainer attribution, software transparency can offer a
building block for future systems to securely track the provenance of code
through various redistribution steps.

\section*{Acknowledgements}
\emph{removed for review}



\printbibliography

\end{document}